\documentclass[12pt,a4papar]{elsarticle}

\usepackage{lineno,hyperref}
\usepackage{graphics,graphicx}
\usepackage{url}
\usepackage{booktabs}
\usepackage{color}
\usepackage{subfigure}
\usepackage{stmaryrd}
\usepackage{wrapfig}
\usepackage{mathrsfs}
\usepackage[justification=centering]{caption}
\usepackage{amsfonts,amsmath,amssymb,amsthm}
\usepackage{epstopdf}

\newtheorem{definition}{Definition}
\newtheorem{theorem}{Theorem}
\newtheorem{corollary}{Corollary}
\newtheorem{remark}{Remark}

\newtheorem{lemma}{Lemma}

\newtheorem{example}{Example}
\newtheorem{property}{Property}
\graphicspath{{./figure/}}
\numberwithin{equation}{section}
\modulolinenumbers[5]
\biboptions{numbers,sort&compress}

\journal{}

\begin{document}

\begin{frontmatter}

\title{Total Positivity of A Kind of Generalized Toric-Bernstein Basis}

\author{Ying-Ying Yu}
\author{Hui Ma}
\author{Chun-Gang Zhu \corref{cor1}}
\cortext[cor1]{Corresponding author.}
\ead{cgzhu@dlut.edu.cn}
\address{School of Mathematical Sciences, Dalian University of Technology, Dalian 116024, China}



\begin{abstract}
 {{The normalized totally positive bases are widely used in many fields. }} Based on the generalized Vandermonde determinant, the normalized total positivity of a kind of generalized toric-Bernstein basis is proved, which is defined on a set of real points. {By this result,  the progressive iterative approximation property of the generalized toric-B\'{e}zier curve is obtained.}
\end{abstract}

\begin{keyword}
totally positive matrix, totally positive basis, Benstein basis function, rational basis function.
\MSC 65D17\sep 15B48 \sep 41A15
\end{keyword}
\end{frontmatter}


\section{Introduction}
Total positivity (TP) is an important and powerful concept which often occurs in many subjects, such as statistics \cite{Goodman1986A,Heiligers1994Totally}, mathematical biology \cite{Mariano1996Total}, combinatorics\cite{Berenstein1996Parametrizations,Brenti1995Combinatorics}, dynamics \cite{Gantmacher2002Oscillation}, approximation theory \cite{Price1968Monotone,CARNICER1994633}, operator theory \cite{Sobolev1975Totally}, and geometry \cite{Sturmfels1988Totally}. In addition, we often see TP in the areas of graph theory, algebraic geometry, stochastic process, game theory, matroid, differential equation and representation theory. In the matrix theory, the mathematician who first considered the TP matrix was I. J. Schoenberg \cite{Schoenberg1930}, and his initial research interest was the number of real roots of polynomials, which pushed him to study the variation diminishing transformation, PF sequence, kernel and spline function. The TP matrix is a matrix whose all minors are non-negative \cite{ANDO1987165}. Using the concept of TP matrix, the normalized totally positive (NTP) basis is defined. The NTP basis is a set of basis functions defined on the parameter domain $\Xi$, which satisfies the properties of non-negative and partition of the unit, meanwhile every collocation matrix of this basis on arbitrary increasing real sequence is a TP matrix \cite{Karlin,Cambridge}. When they are applied in computer aided geometric design (CAGD), many shape preserving properties are obtained, such as convexity preserving property and variation diminishing property \cite{CARNICER1994633,CARNICER1996365}. Lin et. al proved that the curves and surfaces generated by NTP basis have the property of progressive iterative approximation (PIA) \cite{LIN2005575}, which has important applications in offset curve fitting and approximation \cite{Zhang2014} and have been researched for many times.

{ In 2010, Lin \cite{LIN} proposed local PIA method with corresponding iterative format and then proved the convergence. Local PIA has the great flexibility since we can selectively insert data points. A weighted PIA for NTP bases was introduced by Lu \cite{LU} in 2010, meanwhile accelerate the convergence of PIA and gave out the optimal weight. PIA can be applied to part of non-NTP bases by two kinds of generalized PIA proposed by Chen and Wang et. al \cite{Chen} in 2012. In 2018, Zhang et al. \cite{Zhang} studied the geometric iterative method of generalized B-spline curves/surfaces with different weights, and illustrated the flexibility of the method in three aspects. Lin \cite{LIN2} and Lin et al. \cite{LIN3} summarized the geometric iterative methods and their applications.}


In 2002, Krasauskas \cite{Krasauskas2002} proposed a new method for polygonal surface modeling based on toric ideals and toric varieties defined by a given set of integer lattice points, in which basis is called toric-Bernstein (or toric-B\'ezier) basis. When the lattice is restricted to one-dimensional integer points, toric-Bernstein basis degenerates into one-dimensional toric-Bernstein basis, which is the classical Bernstein basis in fact. In 2015, Postinghel et. al \cite{Postinghel2015Degenerations} proposed a class of irrational toric varieties, inspired by which Li et. al \cite{Li,Li1} defined a kind of generalized toric-Bernstein basis ({GT}-Bernstein basis) using a set of real points, and then applied it in the design of curve/surface which is called generalized toric-B\'{e}zier curve/surface ({GT}-B\'{e}zier curve/surface). The main purpose of this paper is to prove that rational {GT}-Bernstein basis is NTP basis.

The rest of paper is organized as follows. In Section 2, we introduce the definition of {GT}-Bernstein basis, rational {GT}-Bernstein basis, and the definition of TP matrix and TP basis. In Section 3, we prove the TP property of rational {GT}-Bernstein basis based on the result of the generalized Vandermonde determinant. Then, in Section 4 we briefly introduce the progressive iterative approximation property of {GT}-B\'{e}zier curve and give some examples. At last, we conclude the paper and propose the work in the future in Section 5.

\section{Preliminaries}

\subsection{Generalized toric-Bernstein basis}
Toric geometry is mainly a theoretical study on toric varieties, toric ideals and related issues developed gradually from the 1970s, and becomes a hot topic of algebraic geometry \cite{Sottile2008Toric,Gronber}. In 2002, Krasauskas \cite{Krasauskas2002} proposed a kind of polygonal surface called toric surface, whose essence is the projection of toric varieties from high-dimensional projective space to low-dimensional affine space \cite{Sottile2008Toric,Gronber}. The basis for toric surface {construction} is defined on a set of integer lattice points and called toric-Bernstein (or toric-B\'ezier) basis. When the lattice is restricted to one-dimensional, the basis degenerates into one-dimensional toric-Bernstein basis, which is the classical Bernstein basis after the parameterization.
\begin{definition}[\cite{Krasauskas2002}]
Let $Q={\{a_{0},\cdots,a_{n}\}}\subset \mathbb{Z}$ be a set of finitely integer points, where $a_{0}< a_{1}< \cdots < a_{n-1}< a_{n}$. For $i\in\{0,1,\cdots,n\}$, we define a basis function corresponding to the point $a_{i}$ of $Q$:
\begin{equation*}
\beta_{a_{i}}(t)=c_{a_{i}}h_{0}(t)^{h_{0}(a_{i})}h_{1}(t)^{h_{1}(a_{i})}, t\in[a_{0},a_{n}],
\end{equation*}
where coefficient $c_{a_{i}}>0$ and
$$h_{0}(t)=t-a_{0}, h_{1}(t)=a_{n}-t.$$
\end{definition}

The set of functions $\{\beta_{a_{i}}(t)\}$ is called toric-Bernstein basis, and we can see that every basis function satisfies $\beta_{a_{i}}(t)\geqslant0, t\in[a_{0},a_{n}]$.

\begin{remark}
If we set $a_{i}=i$ $(i=0,1,\cdots,n)$, then
$$\beta_{a_{i}}(t)=\beta_{i}(t)=c_{i}t^{i}(n-t)^{n-i}, i=0,1,\cdots,n.$$
Let $t=nx$ and $c_{i}=\frac{1}{n^{n}}\binom{n}{i}$, then we have
\begin{equation*}
\beta_{i}(t)=c_{i}t^{i}(n-t)^{n-i}=c_{i}n^{n}x^{i}(1-x)^{n-i}=\binom{n}{i}x^{i}(1-x)^{n-i}=B^{n}_{i}(x)
\end{equation*}
for $i=0,1,\cdots,n$. It can be seen that the toric-Bernstein basis is the generalization of Bernstein basis. And then the NTP property of the basis is determined by that of Bernstein basis directly.
\end{remark}

In 2015, Postinghel et. al proposed a class of irrational toric varieties, which is defined by a set of finitely real points. Inspired by their work, Li et. al \cite{Li,Li1} defined a kind of {GT}-Bernstein basis using a set of real points and then applied it in the curve and surface design.

\begin{definition}[\cite{Li,Li1}]\label{definition:2}
Let $S={\{a_{0},\cdots,a_{n}\}}\subset \mathbb{R}$ be a finite set, where $a_{0}\leqslant a_{1}\leqslant \cdots \leqslant a_{n-1}\leqslant a_{n}$, and $a_{0}<a_{n}$. For any point $a_{i}$ in the set $S$, define the function:
\begin{equation*}
\beta_{a_{i}}(t)=c_{a_{i}}h_{0}(t)^{h_{0}(a_{i})}h_{1}(t)^{h_{1}(a_{i})}, t\in[a_{0},a_{n}],
\end{equation*}
where the coefficient $c_{a_{i}}>0$, and
$$ h_{0}(t)=l_{0}(t-a_{0}), h_{1}(t)=l_{1}(a_{n}-t), l_{0}>0,l_{1}>0.$$
\end{definition}

The set of functions $\{\beta_{a_{i}}(t)\}$ is called {GT}-Bernstein basis. And the function satisfies $\beta_{a_{i}}(t)\geqslant0, t\in[a_{0},a_{n}]$. The points of set $S$ are called nodes.

\begin{remark}
Although {a} {GT}-Bernstein basis $\{\beta_{a_{i}}(t)\}$ depends on the selection of $l_{0},l_{1}$, the {generalized toric-B\'ezier curve (GT-B\'ezier curve)} defined by it depends only on the set $S$ and independent of the selection of $l_{0},l_{1}$ \cite{Li,Li1}. In the rest of the paper, we assume that $l_{0}=l_{1}=l>0$.
\end{remark}

For $a_{i}\in S$, the rational {GT}-Bernstein basis function is defined as:
\begin{equation}
\mathcal{T}_{a_{i}}(t)=\frac{\omega_{a_{i}}\beta_{a_{i}}(t)}{\sum_{i=0}^{n}\omega_{a_{i}}\beta_{a_{i}}(t)}, t\in[a_{0},a_{n}],
\label{align 2.1}
\end{equation}
where $\omega_{a_{i}}>0$ is weight.
Obviously, rational {GT}-Bernstein basis has the properties of non-negative and partition of the unit.

\subsection{Totally positive matrix and totally positive basis}
Let matrix $X=[x_{ij}]$ be a real matrix, where
$$X\left(
\begin{matrix}
i_{1}& i_{2}& \cdots  &i_{k}\\
j_{1}& j_{2}& \cdots  &j_{k}
\end{matrix}
\right)$$
denote the minor matrix of matrix $X$ of order $k$ composed of the $i_{1},i_{2},\cdots,i_{k}$ $({1}\leqslant i_{1}<i_{2}<\cdots<i_{k})$ rows and the $j_{1},j_{2},\cdots,j_{k}$ $({1}\leqslant j_{1}<j_{2}<\cdots<j_{k})$ columns in $X$. Particularly, when $i_{k}=j_{k}={k}$, $X(i,j)$ is the ordinal principal minor of $X$.

\begin{definition}[\cite{Karlin}]\label{definition:3}
An $n\times m$-order real matrix $X$ is called a TP matrix if and only if all of its minors are non-negative, while $X$ is called a {strictly totally positive (STP)} matrix if and only if all of its minors are positive.
\end{definition}

TP matrix has a lot of properties, and many operators can maintain TP property. There are also some methods to construct new TP matrix with TP matrices we have already known. According to the definition of TP matrix and Cauchy-Binet formula, we have:

\begin{property}[\cite{Karlin}]\label{property:1}
Multiply any row (column) of a TP (STP) matrix by a positive scalar, it's still a TP (STP) matrix.
\end{property}

Checking the TP  property of a matrix by Definition \ref{definition:3} is very complicated because the minors of arbitrary order $k$ needed to be calculated. In fact, the TP  property of a matrix is only depended on the TP  property of its minors consisting of the consecutive rows and columns in it.

\begin{theorem}[\cite{Karlin}]
Let $X$ be an $n\times m$-order matrix. Then $X$ is a TP (STP) matrix if all its minors of order $k$ consisting of the consecutive rows and columns in it are TP (STP) matrices, where $k=1,2,\cdots,\min\{n,m\}$.
\end{theorem}

Using the concept of TP matrix, the NTP basis is defined as follow.

\begin{definition}[\cite{LIN2005575}]
The basis $\{B_{0}(t),B_{1}(t),\cdots B_{n}(t)\}$ defined on parameter domain $\Xi\subseteq \mathbb{R}$ is called NTP basis if and only if:
\begin{enumerate}[(i)]
\item The basis is non-negative;
\item The basis is the partition of the unit;
\item The collocation matrix of the basis under an arbitrary set of strictly monotonically increasing real numbers $t_{0}<t_{1}<\cdots<t_{n}$ of $\Xi$ is a TP matrix, that is
\begin{equation*}
\mathcal{\mathcal{M}}(^{B_{0},\cdots,B_{n}}_{t_{0},\cdots,t_{n}})=(B_{j}(t_{i}))^{i=0,1,\cdots,n}_{j=0,1,\cdots,n}
\end{equation*}
is a TP matrix.
\end{enumerate}
\end{definition}

\section{The total positivity of rational generalized toric-Bernstein basis}
 If the points of $S$ are integers, then the {GT-Bernstein} basis is NTP from the NTP property of classical Bernstein basis, which can be applied to prove the PIA property of B\'ezier curves. In this section, we will show that rational {GT}-Bernstein basis defined on arbitrary real points in Definition \ref{definition:2} is NTP basis too, which can be applied for geometric modeling.


\begin{theorem}\label{theorem:2}
The TP property of the collocation matrix
$$B=[\beta_{a_{j}}(t_{i})]^{i=0,1,\cdots n}_{j=0,1,\cdots n}$$
of {GT}-Bernstein basis $\{\beta_{a_{i}}(t)|i=0,1,\cdots,n\}$ on an increasing sequence $a_{0}<t_{0}<t_{1}<\cdots<t_{n}<a_{n}$ is equivalent to the TP property of matrix
$$A=
\left[
\begin{matrix}
x_{0}^{lk_{0}}&  x_{0}^{lk_{1}}&  \cdots   &x_{0}^{lk_{n}}\\
x_{1}^{lk_{0}}&  x_{1}^{lk_{1}}&  \cdots   &x_{1}^{lk_{n}}\\
\vdots&\vdots&\ddots&\vdots\\
x_{n}^{lk_{0}}&  x_{n}^{lk_{1}}&  \cdots   &x_{n}^{lk_{n}}
\end{matrix}
\right],
$$
where $l_{0}=l_{1}=l>0$, $k_{i}=a_{i}-a_{0}$, $i=0,1,\cdots,n$.
\end{theorem}

\begin{proof}
Let the collocation matrix be $B=[B_{0},B_{1},\cdots,B_{n}]_{(n+1)\times(n+1)}$, where
$$B_{j}=
\left[
\begin{matrix}
c_{a_{j}}l_{0}(t_{0}-a_{0})^{l_{0}(a_{j}-a_{0})}l_{1}(a_{n}-t_{0})^{l_{1}(a_{n}-a_{j})}\\
c_{a_{j}}l_{0}(t_{1}-a_{0})^{l_{0}(a_{j}-a_{0})}l_{1}(a_{n}-t_{1})^{l_{1}(a_{n}-a_{j})}\\
\vdots\\
c_{a_{j}}l_{0}(t_{n}-a_{0})^{l_{0}(a_{j}-a_{0})}l_{1}(a_{n}-t_{n})^{l_{1}(a_{n}-a_{j})}\\
\end{matrix}
\right]_{(n+1)\times1}
$$
for $j=0,1,\cdots,n$.

The $(i+1,j+1)-th$ element $\beta_{a_{j}}(t_{i})$ in matrix $B$ can be expressed as :
\begin{align*}
\beta_{a_{j}}(t_{i})&=c_{a_{i}}l_{0}(t_{i}-a_{0})^{l_{0}(a_{j}-a_{0})}l_{1}(a_{n}-t_{i})^{l_{1}(a_{n}-a_{j})}\\
&=c_{a_{i}}l_{0}l_{1}(a_{n}-t_{i})^{l_{1}(a_{n}-a_{0})}(a_{n}-t_{i})^{-l_{1}(a_{j}-a_{0})}(t_{i}-a_{0})^{l_{0}(a_{j}-a_{0})}.
\end{align*}

For $l_{0}=l_{1}=l>0$, we have
\begin{equation*}
\beta_{a_{j}}(t_{i})=c_{a_{i}}l^{2}(a_{n}-t_{i})^{l(a_{n}-a_{0})}\left(\frac{t_{i}-a_{0}}{a_{n}-t_{i}}\right)^{l(a_{j}-a_{0})}.
\end{equation*}

Let $x_{i}=\frac{t_{i}-a_{0}}{a_{n}-t_{i}}$, $k_{i}=a_{i}-a_{0}$, $i=0,1,\cdots,n$, then we have
$$A=
\left[
\begin{matrix}
x_{0}^{lk_{0}}&  x_{0}^{lk_{1}}&  \cdots   &x_{0}^{lk_{n}}\\
x_{1}^{lk_{0}}&  x_{1}^{lk_{1}}&  \cdots   &x_{1}^{lk_{n}}\\
\vdots&\vdots&\ddots&\vdots\\
x_{n}^{lk_{0}}&  x_{n}^{lk_{1}}&  \cdots   &x_{n}^{lk_{n}}
\end{matrix}
\right].
$$

It can be known from the Property \ref{property:1} that the TP property of the matrix $B$ is equivalent to that of the matrix $A$.
\end{proof}

\begin{corollary}\label{corollary:1}
The TP property of the collocation matrix
$$C=[\mathcal{T}_{a_{j}}(t_{i})]^{i=0,1,\cdots n}_{j=0,1,\cdots n}$$
of rational {GT}-Bernstein basis $\{\mathcal{T}_{a_{i}}(t)|i=0,1,\cdots,n\}$ on an increasing sequence $a_{0}<t_{0}<t_{1}<\cdots<t_{n}<a_{n}$ is equivalent to the TP property of matrix A.
\end{corollary}

\begin{proof}
Let $b_{i}=\sum_{j=0}^{n}\omega_{a_{j}}\beta_{a_{j}}(t_{i})>0$, $i=0,1,\cdots,n$. The collocation matrix of rational {GT}-Bernstein basis is $C=[C_{0},C_{1},\cdots,C_{n}]_{(n+1)\times(n+1)}$, where
$$C_{j}=
\left[
\begin{matrix}
\frac{\omega_{a_{j}}c_{a_{j}}l_{0}(t_{0}-a_{0})^{l_{0}(a_{j}-a_{0})}l_{1}(a_{n}-t_{0})^{l_{1}(a_{n}-a_{j})}}{b_{0}}\\
\frac{\omega_{a_{j}}c_{a_{j}}l_{0}(t_{1}-a_{0})^{l_{0}(a_{j}-a_{0})}l_{1}(a_{n}-t_{1})^{l_{1}(a_{n}-a_{j})}}{b_{1}}\\
\vdots\\
\frac{\omega_{a_{j}}c_{a_{j}}l_{0}(t_{n}-a_{0})^{l_{0}(a_{j}-a_{0})}l_{1}(a_{n}-t_{n})^{l_{1}(a_{n}-a_{j})}}{b_{n}}\\
\end{matrix}
\right]_{(n+1)\times1}
$$
for $j=0,1,\cdots,n$.

Extracted the common factor, the TP property of matrix $C$ is equivalent to that of matrix $B$. And then according to Theorem \ref{theorem:2}, the TP property of matrix $C$ is equivalent to that of matrix $A$.
\end{proof}

Let $n\geqslant1$, $\alpha:=(\alpha_{0},\alpha_{1},\cdots,\alpha_{n})$, $t:=(t_{0},t_{1},\cdots,t_{n})^{T}$ be real vectors, and
$$W(t;\alpha):=
\left[
\begin{matrix}
t_{0}^{\alpha_{0}}&&s_{1}t_{0}^{\alpha_{1}}&&s_{2}t_{0}^{\alpha_{2}}&&\cdots&&s_{n}t_{0}^{\alpha_{n}}\\
t_{1}^{\alpha_{0}}&&t_{1}^{\alpha_{1}}&&s_{2}t_{1}^{\alpha_{2}}&&\cdots&&s_{n}t_{1}^{\alpha_{n}}\\
t_{2}^{\alpha_{0}}&&t_{2}^{\alpha_{1}}&&t_{2}^{\alpha_{2}}&&\cdots&&s_{n}t_{2}^{\alpha_{n}}\\
\vdots&&\vdots&&\vdots&&\ddots&&\vdots\\
t_{n}^{\alpha_{0}}&&t_{n}^{\alpha_{1}}&&t_{n}^{\alpha_{2}}&&\cdots&&t_{n}^{\alpha_{n}}
\end{matrix}
\right]
$$
be an $(n+1)\times(n+1)$-order real matrix, where $s_{i}\in\{-1,1\}$, $i=1,2,\cdots,n$.

\begin{lemma}[\cite{YANG2001201}]\label{lemma:1}
The generalized Vandermonde determinant
\begin{equation*}
w(t;\alpha):=w(t_{0},t_{1},\cdots,t_{n};\alpha_{0},\alpha_{1},\cdots,\alpha_{n}):=\det(W(t;\alpha))
\end{equation*}
is positive for
$$\alpha_{n}>\alpha_{n-1}>\cdots>\alpha_{1}>\alpha_{0},t_{n}\geqslant_{s_{n}}t_{n-1}\geqslant_{s_{n-1}}\cdots\geqslant_{s_{2}}t_{1}\geqslant_{s_{1}}t_{0}>0,$$
where
\[\geqslant_{s_{i}}:=\begin{cases}
\geqslant&\text{if $s_{i}=-1$}\\
>&\text{if $s_{i}=1$}
\end{cases}.\]
\end{lemma}

\begin{theorem}
The rational {GT}-Bernstein basis  $\{\mathcal{T}_{a_{i}}(t)\}$ defined on the real nodes $S$ is a NTP basis.
\label{theorem:3}
\end{theorem}

\begin{proof}
It is obvious that rational {GT}-Bernstein basis is non-negative in the domain $\Xi=[a_{0},a_{n}]$, and has the property of partition of the unit. Then we only need to prove the TP property of it.

\begin{enumerate}[(i)]

\item  If $a_{0}<t_{0}<t_{1}<\cdots<t_{n}<a_{n}$, from Corollary \ref{corollary:1}, then we know that the TP property of collocation matrix $C$ is equivalent to the TP property of matrix $A$. We just let $s_{1}=s_{2}=\cdots=s_{n}=1$ in Lemma \ref{lemma:1}. Since $x_{i}=\frac{t_{i}-a_{0}}{a_{n}-t_{i}}$, $k_{i}=a_{i}-a_{0}$, $i=0,1,\cdots,n$, it's easy to know that $x_{i}$ and $k_{i}$ satisfy the conditions in Lemma \ref{lemma:1}, which are $0<x_{0}<x_{1}<\cdots<x_{n}$, and $k_{0}<k_{1}<\cdots<k_{n}$. So $A=\det(W(x;k))>0$, and so does all its minors.
That means matrix $C$ is equivalent to the matrix $A$ as a TP matrix.

\item If $a_{0}{=}t_{0}<t_{1}<\cdots<t_{n}<a_{n}$, then the matrix $C$ is equivalent to
$$A_{1}=
\left[
\begin{matrix}
1&0&\cdots&0\\
x_{1}^{lk_{0}}&  x_{1}^{lk_{1}}&  \cdots   &x_{1}^{lk_{n}}\\
\vdots&\vdots&\ddots&\vdots\\
x_{n}^{lk_{0}}&  x_{n}^{lk_{1}}&  \cdots   &x_{n}^{lk_{n}}
\end{matrix}
\right].
$$

According to Lemma \ref{lemma:1} and calculation rules of determinant, we have $\det(A_{1})>0$, and so does all its minors.

\item If $a_{0}<t_{0}<t_{1}<\cdots<t_{n}{=}a_{n}$, then the matrix $C$ is equivalent to

$$A_{2}=
\left[
\begin{matrix}
x_{0}^{lk_{0}}&  x_{0}^{lk_{1}}&  \cdots   &x_{0}^{lk_{n}}\\
x_{1}^{lk_{0}}&  x_{1}^{lk_{1}}&  \cdots   &x_{1}^{lk_{n}}\\
\vdots&\vdots&\ddots&\vdots\\
0&0& \cdots &1
\end{matrix}
\right].
$$

According to Lemma \ref{lemma:1} and calculation rules of determinant, we obtain $\det(A_{2})>0$, and so does all its minors.

\item If $a_{0}{=}t_{0}<t_{1}<\cdots<t_{n}{=}a_{n}$,  then the matrix $C$ is equivalent to

$$A_{3}=
\left[
\begin{matrix}
1&  0&  \cdots   &0\\
x_{1}^{lk_{0}}&  x_{1}^{lk_{1}}&  \cdots   &x_{1}^{lk_{n}}\\
\vdots&\vdots&\ddots&\vdots\\
0&0& \cdots &1
\end{matrix}
\right].
$$

According to $(\romannumeral2)$ and $(\romannumeral3)$, then $\det(A_{3})>0$, and so does all its minors.
\end{enumerate}

From all of above, we know that every collocation matrix of rational {GT}-Bernstein basis $\{\mathcal{T}_{a_{i}}(t)|i=0,1,\cdots,n\}$ on arbitrary increasing sequence in $\Xi$ is a TP matrix. Thus, rational {GT}-Bernstein basis is a NTP basis.
\end{proof}

{\section{Progressive Iterative Approximation Property of GT-B\'{e}zier Curves}

Let $\{\mathcal{B}_{0}(t),\mathcal{B}_{1}(t),\cdots,\mathcal{B}_{n}(t)\}$ be a set of blending basis defined on parameter domain $\Xi$, $\{\mathbf{P}_{0},\mathbf{P}_{0},\cdots,\mathbf{P}_{n}\}$ be a sequence of points. For each point $\mathbf{P}_{i}$, assign a parameter value $t_{i}\in\Xi,i=0,1,\cdots,n$, which satisfies $t_{0}<t_{1}<\cdots<t_{n}$. Construct an initial curve
\begin{align}
\textbf{C}^{0}(t)=\sum_{i=0}^{n}\mathbf{P}_{i}^{0}\mathcal{B}_{i}(t),\ \ t\in\Xi
\label{align 4.1}
\end{align}
with given control points $\mathbf{P}_{i}^{0}=\mathbf{P}_{i}$.
For $k=0,1,\cdots$, we calculate the $(k+1)-th$ adjustment vectors
$$\Delta_{i}^{k}=\mathbf{P}_{i}-C_{i}^{k}(t_{i}),\ \ i=0,1,\cdots,n, $$
and then let
$$\mathbf{P}_{i}^{k+1}=\mathbf{P}_{i}^{k}+\Delta_{i}^{k},\ \ i=0,1,\cdots,n.$$
Thus, we have the $(k+2)-th$ iterative curve
$$\textbf{C}^{k+1}(t)=\sum_{i=0}^{n}\mathbf{P}_{i}^{k+1}\mathcal{B}_{i}(t),\ \ t\in\Xi.$$

Keep doing this iteration process, we can get $\{\textbf{C}^{k}(t)|k=0,1,\cdots\}$ as sequence of curves. The initial curve \eqref{align 4.1} has the PIA property \cite{LIN2005575} if it satisfies
$$\lim_{k\rightarrow\infty}\textbf{C}^{k}(t_{i})=\mathbf{P}_{i}^{0}$$
for $i=0,1,\cdots,n$.

In \cite{LIN2005575}, Lin et al. proposed a simple method to identify the PIA property of a curve by its blending basis.
\begin{theorem}\cite{LIN2005575}
The curve (\ref{align 4.1}) has the progressive iterative approximation property if its blending basis $\{\mathcal{B}_{i}(t)\}$ is normalized totally positive basis.
\label{theorem:4}
\end{theorem}

 In \cite{Li}, a kind of parametric curve is constructed by rational GT-Bernstein basis (\ref{align 2.1}).
\begin{definition}\cite{Li}
The curve
\begin{align}
\textbf{P}(t)&=\sum_{i=0}^{n}\mathbf{P}_{a_{i}}\frac{\omega_{a_{i}}\beta_{a_{i}}(t)}{\sum_{i=0}^{n}\omega_{a_{i}}\beta_{a_{i}}(t)}=\sum_{i=0}^{n}\mathbf{P}_{a_{i}}\mathcal{T}_{a_{i}}(t), \ \ t\in[a_{0},a_{n}]
\label{align 4.2}
\end{align}
is called generalized toric-B\'{e}zier curve (GT-B\'{e}zier curve) of degree $n$, where $\{\mathbf{P}_{a_{0}},\mathbf{P}_{a_{1}},\cdots\mathbf{P}_{a_{n}}\}\subset \mathbb{R}^{3}$ are control points corresponding to nodes $S=\{a_{0},a_{1},\cdots,a_{n}\}$.
\end{definition}

From \cite{Li}, we know that GT-B\'{e}zier curve preserves some properties with B\'{e}zier curve, such as endpoint interpolation, geometric invariance, affine invariance and so on. According to Theorem \ref{theorem:3}, rational GT-Bernstein basis (\ref{align 2.1}) is NTP basis, and this leads to the PIA property of the GT-B\'{e}zier curve by Theorem \ref{theorem:4}.

\begin{corollary}
GT-B\'{e}zier curve (\ref{align 4.2}) has progressive iterative approximation property.
\end{corollary}
Next, we will illustrate two examples to show the PIA property of GT-B\'{e}zier curve.

\begin{example}
We sample five points $\{\mathbf{P}_{0},\mathbf{P}_{1},\cdots,\mathbf{P}_{4}\}$ from the circle $(x(t),y(t))=(cos(t),sin(t))$ as
 $$\mathbf{P}_{i}=(x(\xi_{i}),y(\xi_{i})), \ \ \xi_{i}=i\times\frac{\pi}{4},i=0,1,2,4, \xi_{3}=\pi\times\frac{\pi}{4}.$$
 We set
 $$t_{i}=i\times\frac{\pi}{4},\ \ i=0,1,2,4,\ \  t_{3}=\pi\times\frac{\pi}{4}.$$
 As shown in Figure \ref{subfig1:a}, the GT-B\'{e}zier curve (red solid), B\'{e}zier curve (blue dash) and rational B\'{e}zier curve (green dot) are selected as initial curves. For GT-B\'{e}zier curve, we set $S=\{a_{i}=t_i|i=0,1,\cdots4\}$, $\{c_{a_{i}}|i=0,1,\cdots4\}=\{1,0.9,0.8,0.9,1\}$ and $l=4.5$. The GT-B\'{e}zier curve and rational B\'{e}zier curve are assigned with the same weights $\{0.5,2.51,5.5,2.51,0.22\}$. Figure \ref{subfig1:b}-\ref{subfig1:d} show the effects of these three curves with different iterations and the Table \ref{table 1} show the iteration errors.

\begin{figure}[!h]
\begin{center}
\subfigure[Initial curves]{
\label{subfig1:a}
\includegraphics[width=6cm]{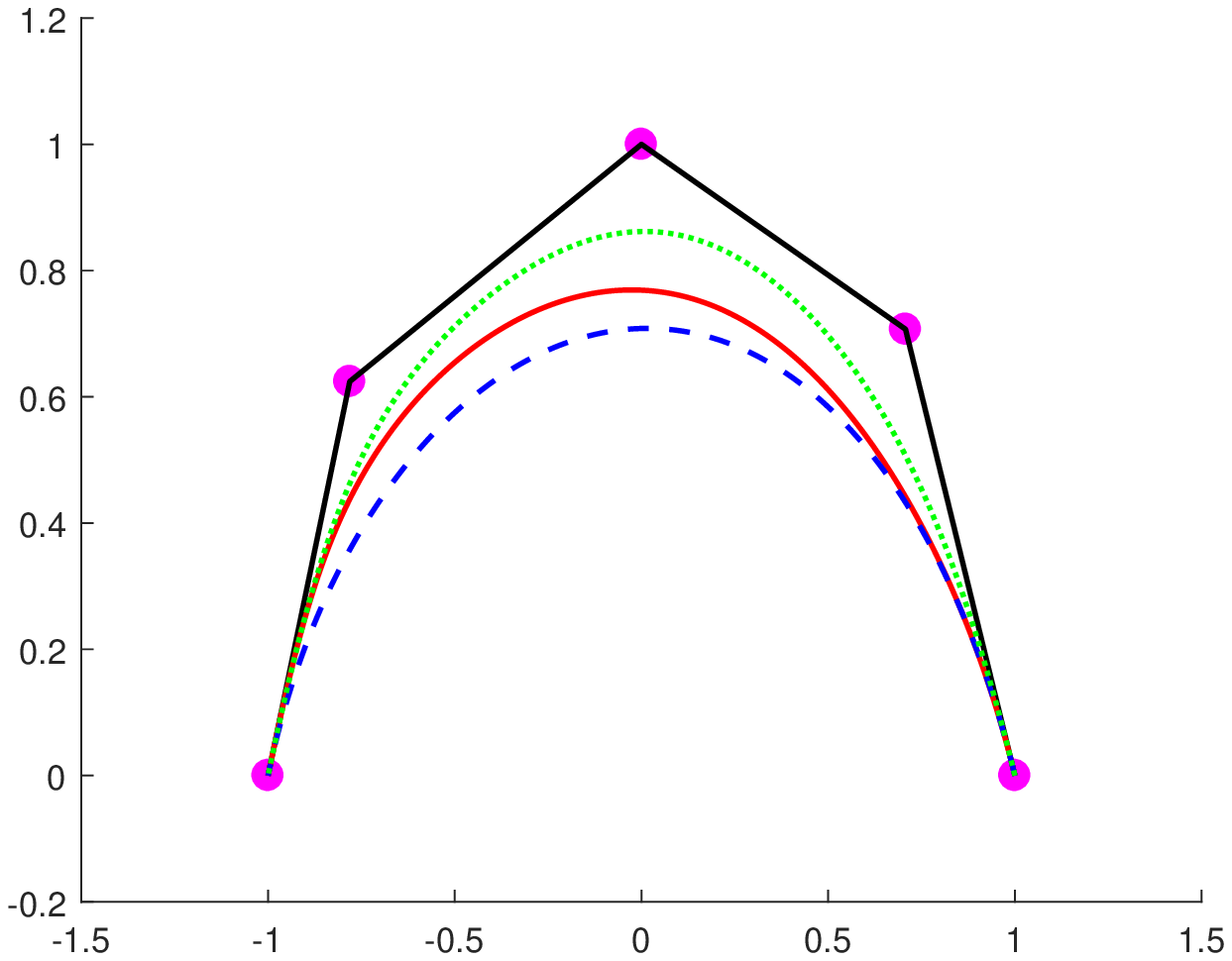}}
\hspace{0.5cm}
\subfigure[5 iterations]{
\label{subfig1:b}
\includegraphics[width=6cm]{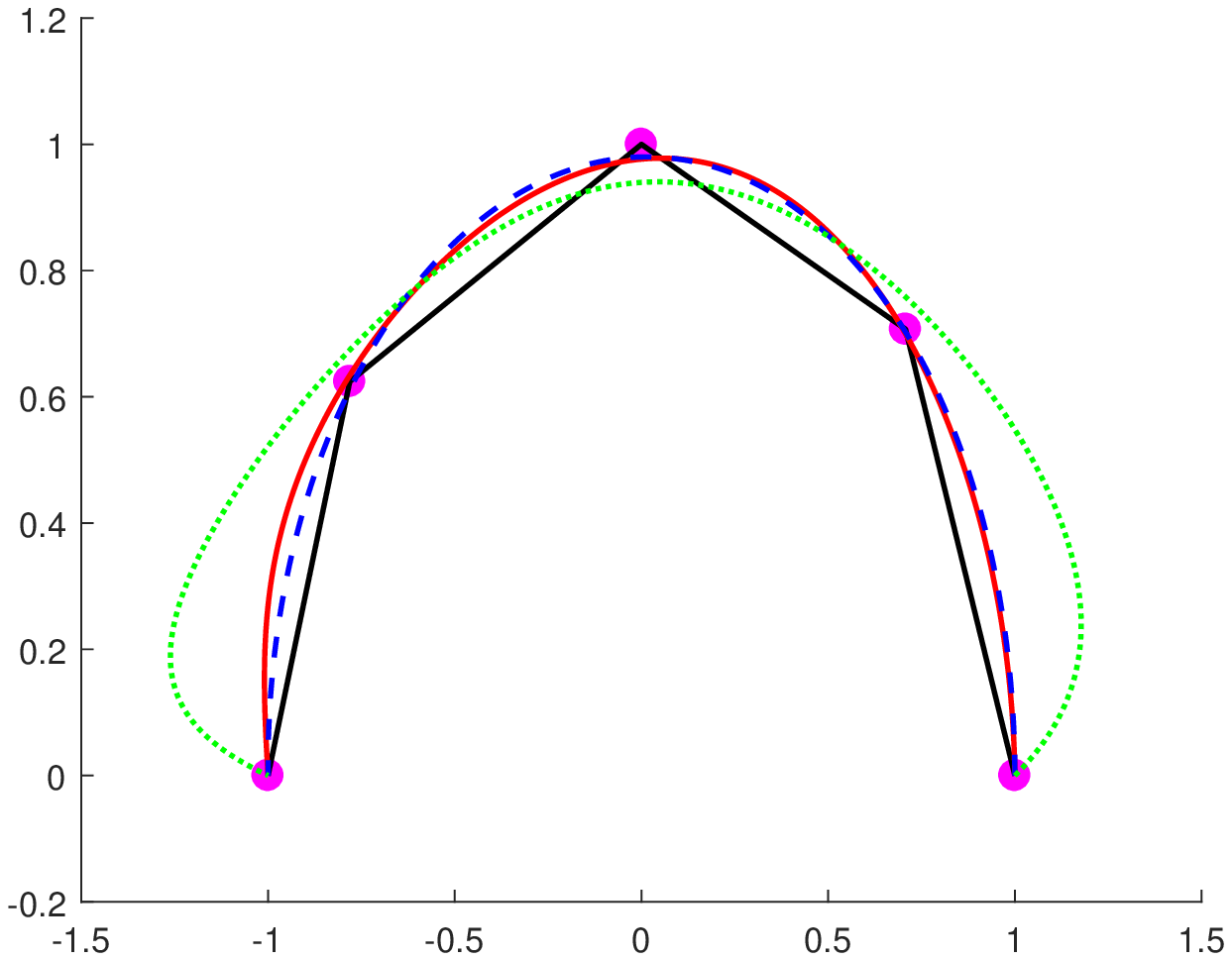}}
\subfigure[10 iterations]{
\label{subfig1:c}
\includegraphics[width=6cm]{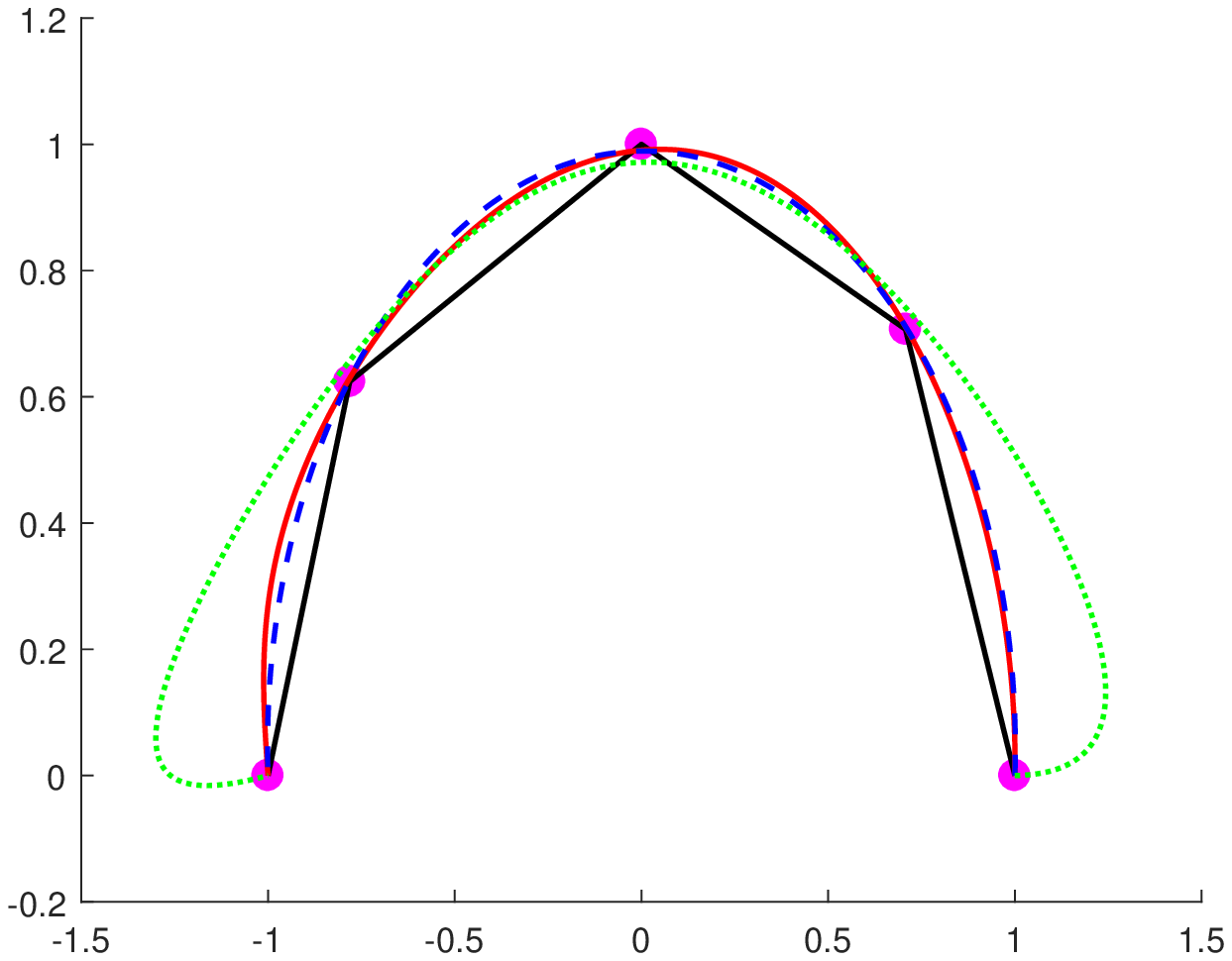}}
\hspace{0.5cm}
\subfigure[20 iterations]{
\label{subfig1:d}
\includegraphics[width=6cm]{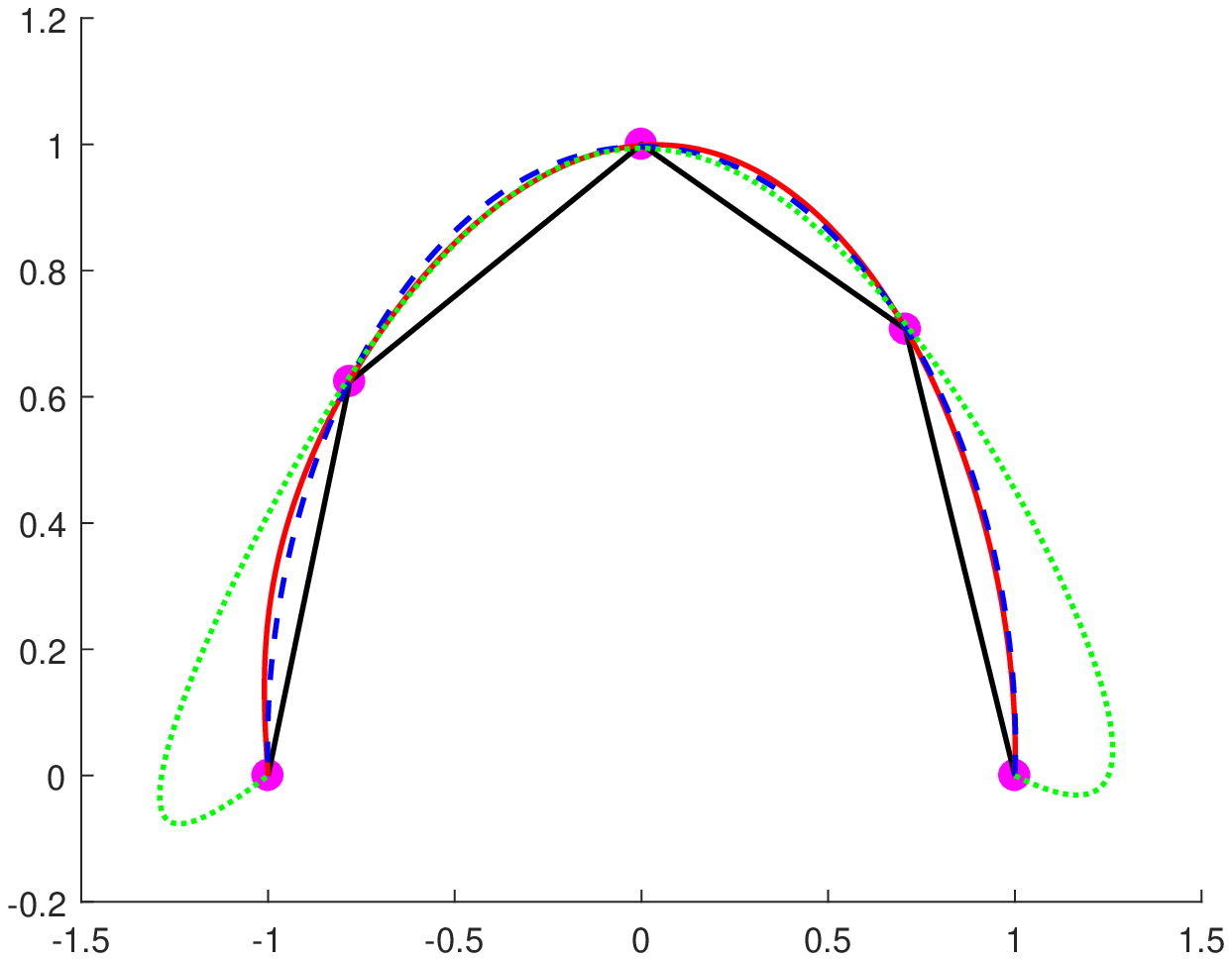}}
\caption{PIA effects of three curves with different iterations}
\label{figure 1}
\end{center}
\end{figure}

\begin{table}[!h]
\begin{center}
\setlength{\tabcolsep}{3mm}{
\begin{tabular}{ccccc}
\toprule
$Iterations$&$1$&$5$&$10$&$20$\\
\midrule
GT-B\'{e}zier curve&$2.317e-01$&$2.236e-02$&$9.7e-03$&$1.8e-03$\\
\midrule
B\'{e}zier curve&$2.927e-01$&$2.228e-02$&$1.08e-02$&$3.6e-03$\\
\midrule
Rational B\'{e}zier curve&$4.134e-01$&$9.84e-02$&$4.16e-02$&$9.4e-03$\\
\bottomrule
\end{tabular}}
\end{center}
\caption{\label{table 1}Errors of three curves with different iterations}
\end{table}
\end{example}

\begin{example}
We sample 31 points $\{\mathbf{P}_{i}=(x(\xi_{i}),y(\xi_{i}),z(\xi_{i}))|i=0,1,\cdots,30\}$ from the parametric curve
$$(x(t),y(t),z(t))=\left(cos(\pi t),sin(\pi t),\frac{t}{6}\right), \ \ t\in[0,2\pi],$$
where
$$\xi_{i}=i\times\frac{2\pi}{30},\ \ i=0,1,\cdots,30\ i\neq 3,\ \ \xi_{3}=\pi\times\frac{2\pi}{30}.$$
We choose
$$t_{i}=i\times\frac{2\pi}{30},\ \ i=0,1,\cdots,30,\ i\neq3,\  \ t_{3}=\pi\times\frac{2\pi}{30}$$
and select GT-B\'{e}zier curve(red solid) and B\'{e}zier curve(blue dot) as initial curves (see Figure \ref{subfig:a}). For GT-B\'{e}zier curve, we set $S=\{a_{i}=t_i|i=0,1,\cdots,30\}$, $\{\omega_{a_{i}}=\binom{31}{i}|i=0,1,\cdots,30\}$, $\{c_{a_{i}}=\frac{1}{30^{2}}|i=0,1,\cdots,30\}$, and $l=31.1$. The effects of
these two curves with different iterations are shown in Figure \ref{subfig:b}-\ref{subfig:d} and the iteration errors are shown in Table \ref{table 2}.

\begin{figure}[!h]
\begin{center}
\subfigure[Initial curves]{
\label{subfig:a}
\includegraphics[width=6cm]{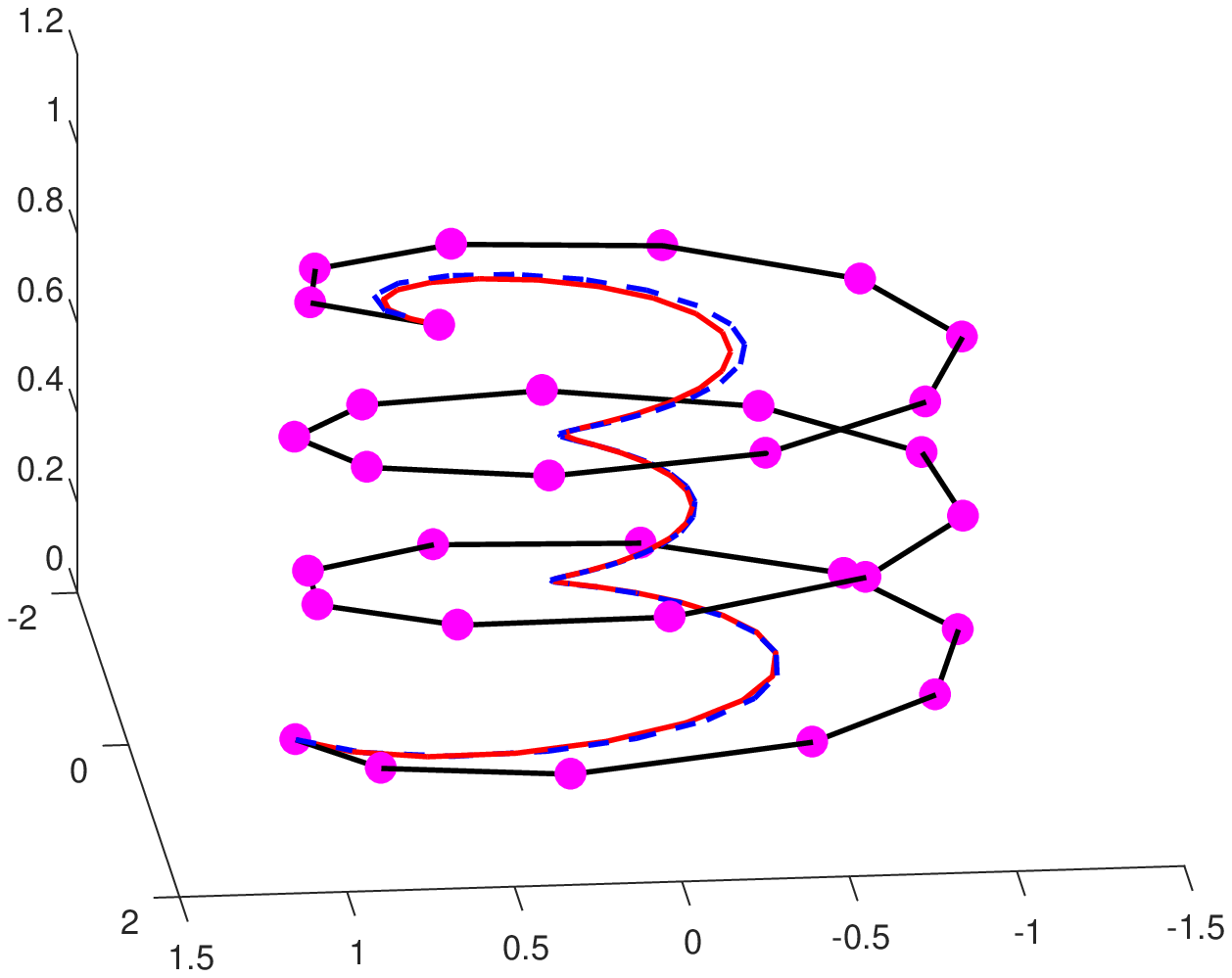}}
\hspace{0.5cm}
\subfigure[10 iterations]{
\label{subfig:b}
\includegraphics[width=6cm]{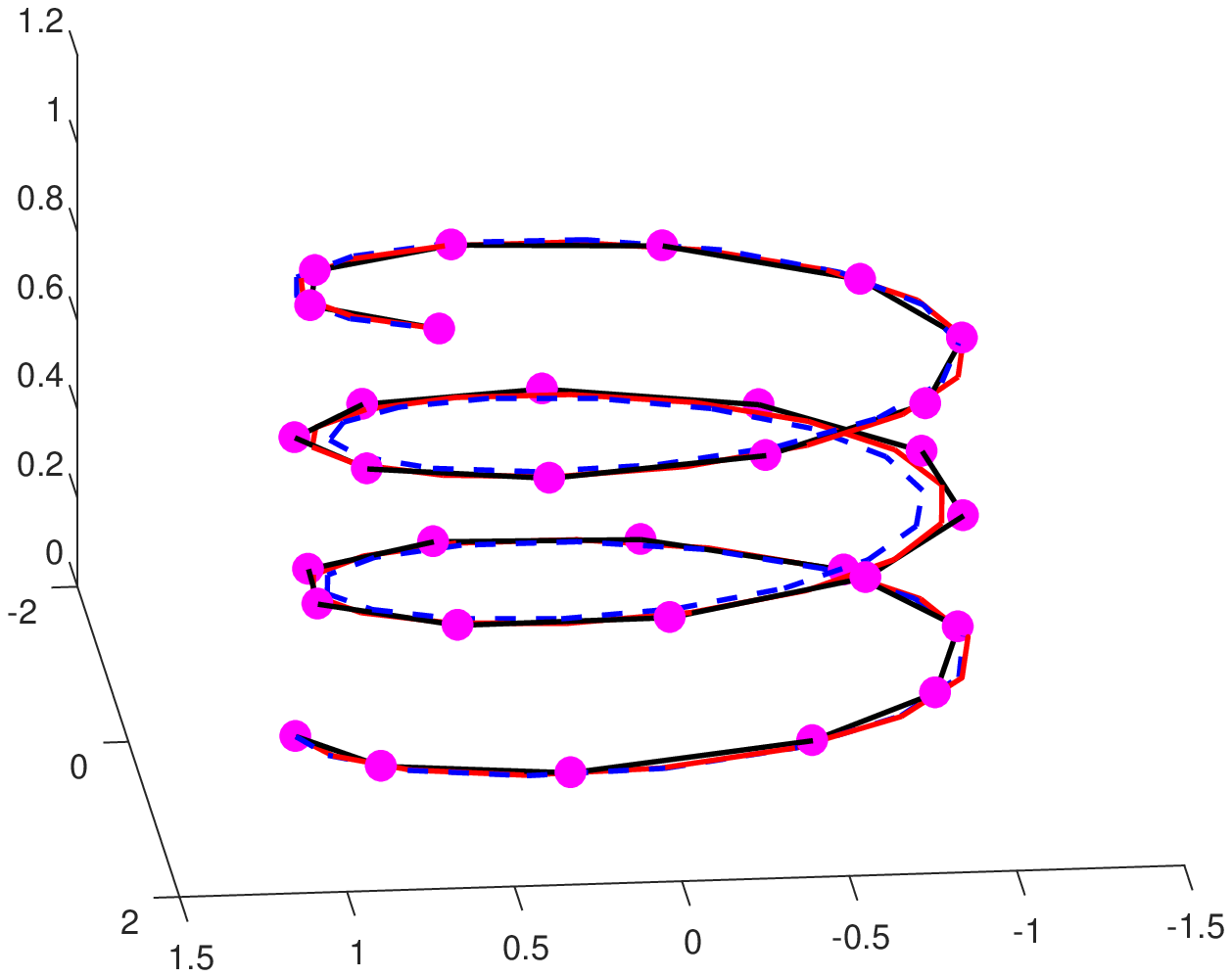}}
\subfigure[20 iterations]{
\label{subfig:c}
\includegraphics[width=6cm]{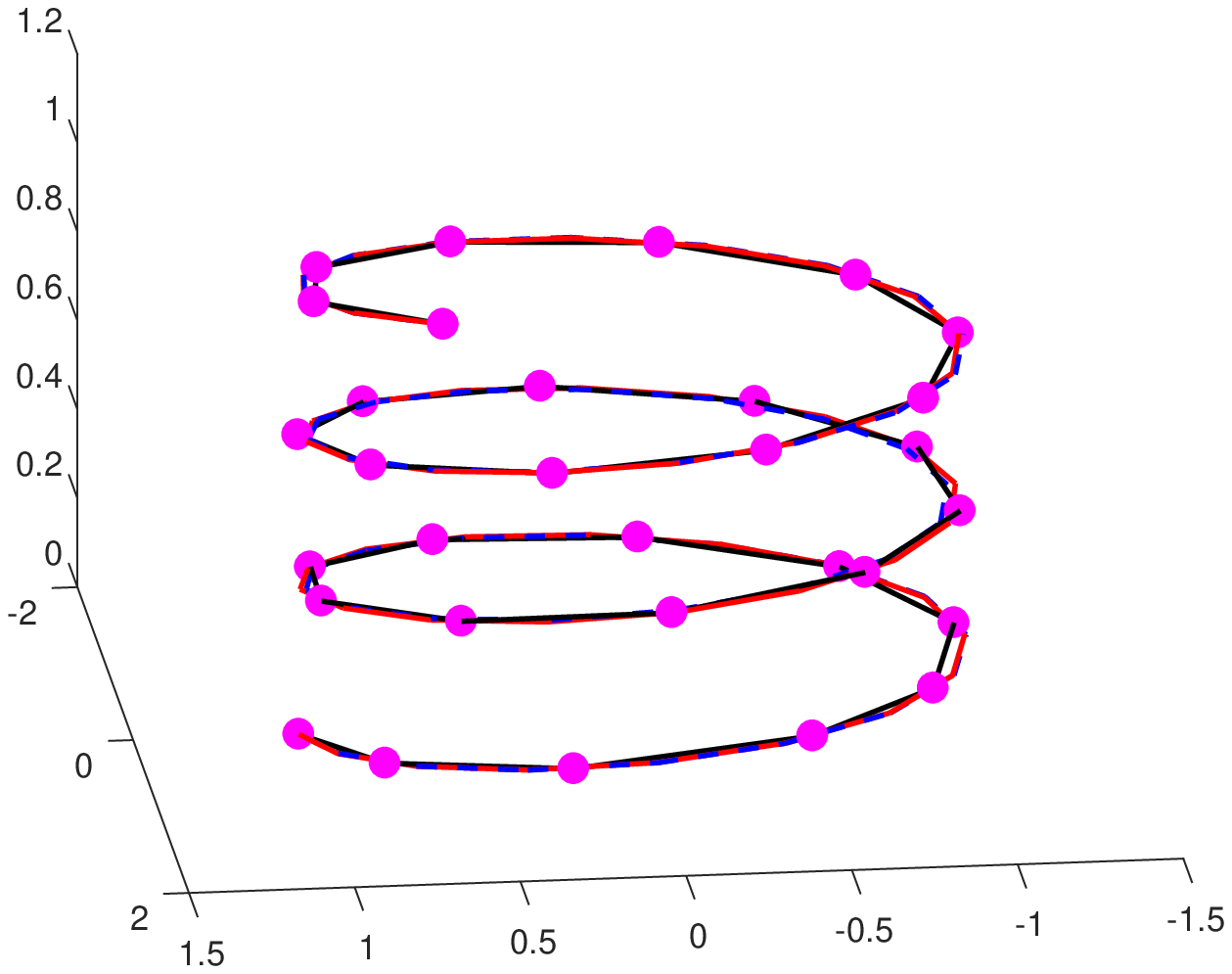}}
\hspace{0.5cm}
\subfigure[30 iterations]{
\label{subfig:d}
\includegraphics[width=6cm]{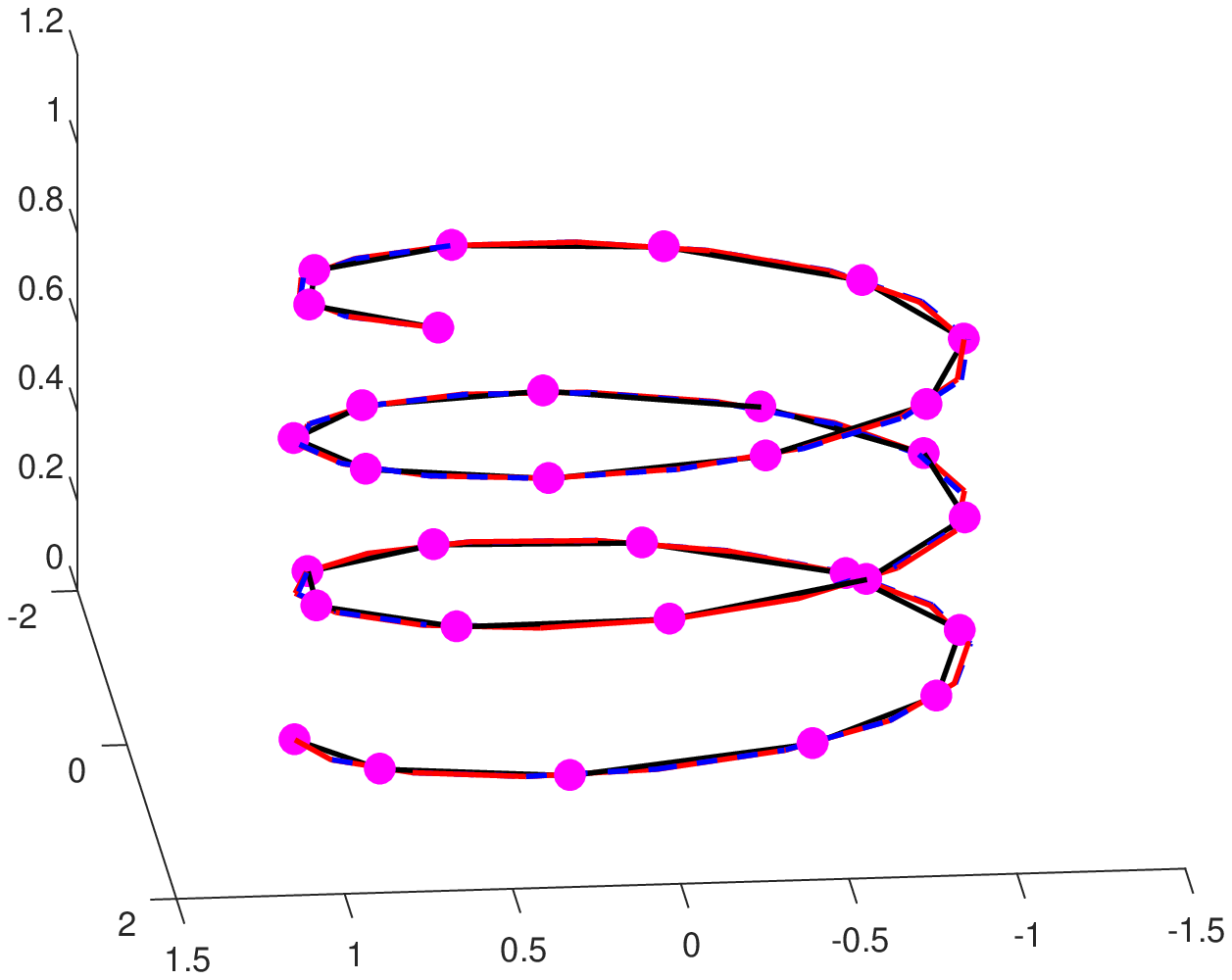}}
\caption{PIA effects of two curves with different iterations}
\label{figure 2}
\end{center}
\end{figure}

\begin{table}[!h]
\begin{center}
\setlength{\tabcolsep}{3mm}{
\begin{tabular}{ccccc}
\toprule
$Iterations$&$1$&$10$&$20$&$30$\\
\midrule
GT-B\'{e}zier curve&$8.390e-01$&$8.92e-02$&$1.90e-02$&$8.7e-03$\\
\midrule
B\'{e}zier curve&$8.086e-01$&$1.245e-01$&$3.86e-02$&$1.84e-02$\\
\bottomrule
\end{tabular}}
\end{center}
\caption{\label{table 2}Errors of two curves with different iterations}
\end{table}

\end{example}

}


\section{Conclusions}
Based on the result of the generalized Vandermonde determinant defined on the set of real points, we prove that the collocation matrix of rational {GT}-Bernstein basis is a TP matrix, and then the rational {GT}-Bernstein basis is NTP basis defined on the finitely real points
$$S={\{a_{0},\cdots,a_{n}\}}\subset \mathbb{R}.$$
{This means that the GT-B\'{e}zier curve defined has the PIA property.

In \cite{Li1}, they {also construct the so called} GT-B\'{e}zier surface by {two dimensional} GT-Bernstein basis.} We conjecture that the rational toric-Bernstein basis defined on two (or higher)-dimensional real points is also NTP basis, and then {GT}-B\'{e}zier surface also has the PIA property. To prove that is our future work.

\section*{Acknowledgements}
{The authors are grateful to the reviewers for their helpful comments and suggestions.}
 This work is partly supported by the National Natural Science Foundation of China (Nos. 11671068, 11271060).

\section*{References}
\bibliographystyle{elsarticle-num}

\end{document}